\documentclass{article}

\usepackage[margin=1in]{geometry}
\usepackage{authblk}

\usepackage[
  colorlinks,
  linkcolor = blue,
  citecolor = blue,
  urlcolor = blue]{hyperref}
\usepackage{amsmath,amssymb,amsthm}
\usepackage{mathtools}
\mathtoolsset{centercolon}

\newcommand{\N}{\mathbb{N}}
\newcommand{\C}{\mathbb{C}}

\DeclarePairedDelimiter{\set}{\lbrace}{\rbrace}
\DeclarePairedDelimiter{\abs}{\lvert}{\rvert}

\newcommand{\tp}{^{\textrm{T}}}

\newcommand{\id}{I}
\newcommand{\ie}{\textit{i.e.}}
\newcommand{\cf}{\textit{cf. }}
\newcommand{\eg}{\textit{e.g.}}
\newcommand{\tg}{\tilde{G}}

\def\be{\begin{equation}}
\def\ee{\end{equation}}

\newcommand{\la}{\langle}
\newcommand{\ra}{\rangle}

\newcommand{\al}{\alpha}

\newcommand{\nong}{G=(V,\pi)}
\newcommand{\M}{\mathcal{M}}
\newcommand{\si}{\psi\psi^*}

\newcommand{\mcs}{\mathcal{S}}

\DeclareMathOperator{\tr}{Tr}

\DeclareMathOperator{\supp}{supp}

\DeclareMathOperator{\pr}{Pr}
\DeclareMathOperator{\rank}{rank}

\newcommand{\Thm}[1]{\hyperref[thm:#1]{Theorem~\ref*{thm:#1}}}
\newcommand{\Lem}[1]{\hyperref[lem:#1]{Lemma~\ref*{lem:#1}}}
\newcommand{\Cor}[1]{\hyperref[cor:#1]{Corollary~\ref*{cor:#1}}}
\newcommand{\Def}[1]{\hyperref[def:#1]{Definition~\ref*{def:#1}}}
\newcommand{\Obs}[1]{\hyperref[obs:#1]{Observation~\ref*{obs:#1}}}
\newcommand{\Rem}[1]{\hyperref[obs:#1]{Remark~\ref*{rem:#1}}}

\newcommand{\Sect}[1]{\hyperref[sec:#1]{Section~\ref*{sec:#1}}}
\newcommand{\Apx}[1]{\hyperref[sec:#1]{Appendix~\ref*{apx:#1}}}

\newcommand{\Fig}[1]{\hyperref[fig:#1]{Figure~\ref*{fig:#1}}}
\newcommand{\Tab}[1]{\hyperref[tab:#1]{Table~\ref*{tab:#1}}}
\newcommand{\EqRef}[1]{\hyperref[eq:#1]{(\ref*{eq:#1})}}
\newcommand{\Eq}[1]{Equation~\hyperref[eq:#1]{(\ref*{eq:#1})}}

\newcommand{\PERFECT}{\texttt{PERFECT-PME}}
\newcommand{\PERFECTC}{\texttt{PERFECT-SYN}}
\newcommand{\QIND}{\texttt{Q-INDEPENDENCE}}
\newcommand{\per}{\texttt{PERFECT}}
\newcommand{\pme}{\texttt{PME}}

\newtheorem{theorem}{Theorem}[section]

\newtheorem{lemma}[theorem]{Lemma}

\newtheorem{rem}[theorem]{Remark}

\theoremstyle{definition}
\newtheorem{definition}[theorem]{Definition}

\title{Deciding the existence of perfect entangled strategies for nonlocal games}

\author[1]{Laura Man\v{c}inska\thanks{laura@locc.la}}
\author[1,2]{David E.~Roberson\thanks{davideroberson@gmail.com}}
\author[1,2]{Antonios Varvitsiotis\thanks{AVarvitsiotis@ntu.edu.sg}}
\affil[1]{Centre for Quantum Technologies, National University of Singapore, 117543 Singapore}
\affil[2]{School of Physical and Mathematical Sciences, Nanyang Technological University, 50 Nanyang Avenue, 637371 Singapore}

\begin{document}
\maketitle

\begin{abstract}
\noindent First, we consider the problem of deciding  whether  a nonlocal game admits  a perfect entangled  strategy that uses  projective measurements on a maximally entangled shared state. Via a polynomial-time Karp reduction, we show that independent set games are the hardest instances of this problem. Secondly, we show that  if  every  independent set game  whose entangled value is equal to one  admits  a perfect entangled strategy, then   the same holds for all symmetric synchronous  games.  Finally, we identify combinatorial lower bounds on the classical and entangled values of synchronous games in terms of variants of the independence number of appropriate graphs. Our results suggest that independent set games might be representative of all nonlocal games when dealing with questions concerning  perfect entangled strategies.
\medskip

\end{abstract}

\section{Introduction}
Entanglement plays a central role in quantum information processing and is  increasingly  seen as a valuable resource for distributed tasks such as unconditionally secure cryptography~\cite{Ekert91}, randomness certification~\cite{Colbeck09,Pironio} and expansion~\cite{VV12,CoudronY13infinite}.   Given such a scenario  it is interesting to understand how much and what kind of entanglement needs to be employed in an optimal  entangled strategy. As is commonly done, we study these questions within the framework of nonlocal games. In the  computer science community nonlocal games arise as one-round interactive proof systems, while in the physics community they are known as Bell inequalities~\cite{Brunner14}.

A {\em nonlocal  game}   is specified by four  finite sets $A,B,Q,R$, a  probability distribution $\pi$ on $Q\times R$ and a Boolean predicate $V: A\times B\times Q\times R\rightarrow \set{0,1}$.  The game proceeds as follows: Using $\pi$  the verifier samples a pair  $(q,r)\in Q\times R$  and sends $q$ to Alice and $r$ to Bob. Upon receiving their questions the players   respond with  $a\in A$ and $b\in B$,  respectively. The players have knowledge of  the distribution $\pi$ and the predicate $V$ and can agree on a common strategy before the start of the game, but they are not allowed to  communicate after they receive their questions.   We say the players {\em win} the game   if   $V(a,b|q,r)=1$. A strategy is called {\em perfect} if it allows the players to  win the game with probability one.

The goal of Alice and Bob  is to maximize their probability of winning the game.
The {\em classical value} of a game $G$, denoted $\omega(G)$,  is the maximum expected winning probability when the players use deterministic  strategies. 
An entangled strategy for a nonlocal game  allows the players to determine their answers  by  performing joint  measurements on a shared finite-dimensional entangled state. The {\em entangled value} of a game $G$, denoted  $\omega^*(G)$, is the supremum expected winning probability the players can achieve using entangled strategies.

Despite significant efforts, many fundamental questions concerning the properties of the entangled value have so far remained beyond reach:
\begin{itemize}
\item[(i)] {\bf The computability question:} Determine (or upper bound) the computational complexity of $\omega^*(G)$.
\item[(ii)]{\bf The attainability question:} Determine if $\omega^*(G)$ can always be attained.
\item[(iii)]{\bf The resources question:} How much and what kind of entanglement is needed to achieve $\omega^*(G)$.
\end{itemize}

The above questions are understood only for some very special classes of games. One notable example is the class of XOR games; for these games the answer sets, $A$ and $B$, are binary and the verification predicate only depends on the XOR of the player's answers. 
For XOR games  the entangled value  can be formulated   as  a semidefinite program which can be approximated within arbitrary precision in polynomial time. Furthermore,  the entangled value of an XOR game is always attained by a maximally entangled state~\cite{Cleve04}. 

At present, there has only been sporadic progress for other classes of nonlocal games.  Some positive approximability  results have been derived for the class of unique nonlocal games \cite{KRT09}.
For general nonlocal games a hierarchy of semidefinite programming upper bounds for the entangled value was identified in \cite{NPA07}. Unfortunately, the quality of the approximation at each level of the hierarchy  is  not understood.

 Given the lack of   progress in addressing  these questions there has been increasing interest in the  study of  restricted variants of the above problems. A decision problem that  has gained some attention is the following: 

\medskip 
\noindent{\texttt{PERFECT}}\\ 
\indent\textit{Instance:} \ A nonlocal game $G$.\\
\indent\textit{Question:}  Does $G$ admit a perfect entangled  strategy? \medskip

 It follows from recent  work of  Ji that $\per$ is NP-hard \cite{Ji13}.  On the other hand, despite  significant efforts $\per$   is currently not   known  to be   decidable.
Some partial progress has been documented  concerning the decidability of $\per$  for special classes of nonlocal games.  Specifically, Cleve and Mittal have shown  that for  BCS  games, deciding the existence of a perfect entangled strategy can be reduced  to  deciding the existence of a self-adjoint  operator solution to a polynomial system in non-commuting variables \cite{Cleve12}. This reduction does not imply decidability  since  no algorithms are currently  known for deciding the existence of operator solutions to non-commutative polynomial systems. In a follow up  work  Arkhipov  studies parity BCS games with the additional requirement  that each variable appears in  exactly  two clauses.  To any such game he associates an undirected graph and shows that the game   has  a perfect entangled strategy if and only if the corresponding graph is non-planar~\cite{Arkhipov12}. Since non-planarity can be decided in linear time \cite{HT74} this shows that  $\per$ can also  be decided in linear time 
for this special subclass of BCS games.

\paragraph{Motivation, results, and discussion.}
In view of the limited   progress in understanding the computability and attainability questions and with the hope to gain new insights, in this work we study the decision problem  $\per$ where  we impose additional  operational restrictions on the set of allowed strategies. Specifically, our  goal is  to decide whether a given nonlocal game $G$  admits a perfect entangled strategy where the players are only allowed  to apply {\em projective measurements} on a shared  {\em maximally entangled} state (hereafter abbreviated  as $\pme$ strategies).  
Formally,  we focus on  the following decision problem: 

\medskip 
\noindent{\PERFECT}\\
\indent\textit{Instance:} \ A nonlocal game $G$.\\
\indent\textit{Question:}  Does $G$ admit a perfect $\pme$   strategy?
\medskip 

\noindent The study of $\PERFECT$ is motivated by the following considerations. Firstly, given the impasse  on the general  question of deciding whether a nonlocal game admits a perfect strategy, $\PERFECT$ can be viewed as an even further restricted variant of $\per$ that can hopefully provide useful  insights into   the general problem. Secondly, 
to the best of our knowledge,  there are no known examples of nonlocal games that admit perfect strategies but cannot be won perfectly  using $\pme$  strategies.
Thus, $\pme$ strategies  might even be sufficient to reach success probability one (whenever this can be done using some quantum strategy) which would imply  that  $\per$   is in fact  equivalent to  $\PERFECT$.  

We note that the situation is quite   different when one considers non-perfect strategies. Specifically, there are   examples of nonlocal games whose entangled value is strictly smaller than  one and  for which maximally entangled states do not suffice  to achieve the optimal success probability, e.g.,  \cite{Junge11,Vidick10,Liang11,Regev12}.

The decision problem $\PERFECT$ has  also been  considered  by Ji~\cite{Ji13}. Similarly to \cite{Cleve12},  in this work Ji  shows  how one can associate to any nonlocal game $G$  a polynomial system in non-commuting operator variables  with the property that $G$  admits a perfect  $\pme$ strategy if and only if the corresponding system has   a solution in self-adjoint operator variables. Using this reduction he proceeds to show  that $\PERFECT$ is NP-hard already when  the input is restricted to be the BCS game corresponding to the 3-SAT problem.

Our main   result in this work is given in \Thm{main2} where we identify independent set games as being the {\em hardest instances} of $\PERFECT$. In the  {\em $(X,t)$-independent set game}  the players aim to convince a verifier that a graph $X$ contains  an independent set of size $t$ (\ie, a set of $t$ pairwise nonadjacent vertices). To play the game the verifier selects uniformly at random a pair of indices  $(i,j)\in [t]\times [t]$ and sends $i$ to Alice and $j$ to Bob. The players respond with vertices  $u,v \in V(X)$ respectively.   In order to win, the players need to respond with the same vertex of $X$ whenever they receive the  same index. Furthermore, if   they receive  $i\ne j\in [t]$  they need to respond with nonadjacent (and distinct) vertices of $X$.  The second decision problem relevant to this work is $\per$ where the  input is  restricted to be an independent set game.  

\medskip 
\noindent{\QIND}\\
\indent\textit{Instance:} An $(X,t)$-independent set game.\\
\indent\textit{Question:} Does the game admit a perfect entangled  strategy?
\medskip 

In our  main result given in \Thm{main2} we show that any instance $G$ of $\PERFECT$ can be transformed   in polynomial-time to an instance $G'$  of   \QIND{} with the property  that $G$ admits  a perfect $\pme$ strategy if and only if $G'$ admits  a perfect strategy. Formally: 

\medskip 
\noindent\textbf{Result 1:} \emph{{\rm \PERFECT{}} is polynomial-time (Karp) reducible to  {\rm \QIND}.}
\medskip

It is known that $\pme$ strategies  suffice to win  independent set games perfectly  (whenever this can be done using some quantum strategy) \cite{Roberson14}. 
As a result,  all instances of \QIND{} can be identified with instances of \PERFECT{} 
and thus our first result  can be understood as identifying     \QIND{}  to be   among  the most expressive subproblems of \PERFECT.

As an immediate consequence  of our first result  and the discussion in the previous paragraph  it follows that   \PERFECT{}  is decidable if and only if  $\QIND$ is decidable. Currently, it is not known whether    \QIND{}  is   decidable. Nevertheless,   reducing   the decidability question from arbitrary games to the special  class of independent set games allows to narrow down our  focus to this specific class of games for which it might be easier to  make progress on the decidability question.

The proof of  Result 1 consists of   two steps which we now briefly describe. We first need to introduce some definitions. A nonlocal game is called {\em synchronous} if it satisfies the  following three requirements: \textit{(i)} Alice and Bob share the same question set $Q$ and answer set $A$, \textit{(ii)} $ \pi(q,q)>0$ for all $q\in Q$,  and  \textit{(iii)} $V(a,b|q,q)=0$ for all $q\in Q$ and all $a\neq b$.
Notice that the  $(X,t)$-independent set game defined above  is an example of a synchronous nonlocal game. The second decision problem of  interest  in this paper is a variation of $\per$ where the  input is restricted to be a   synchronous  game: 

\medskip 
\noindent{\PERFECTC}\\
\indent\textit{Instance:} A synchronous nonlocal game $G$.\\
\indent\textit{Question:}  Does $G$ admit a perfect quantum strategy? 
\medskip 

In \Lem{ProjStrat} we show that  any  synchronous game that admits a perfect quantum strategy also has a perfect $\pme$ strategy.    Notice that this implies that \PERFECTC{} is a subproblem of \PERFECT.

The first step in proving  Result 1 is  \Lem{GtoTG} where we  show that  \PERFECT{} is polynomial-time reducible to   \PERFECTC{}.  To achieve this  we  extend any  nonlocal game $G$   to a synchronous game $\tilde{G}$ where  we can also ask Alice any of Bob's questions and vice versa (see \Def{corrextension}). The extended game $\tilde{G}$ has the property that $G$ has a perfect $\pme$ strategy if and only if $\tilde{G}$ has a perfect strategy.

The second step in proving  Result 1  is  \Lem{alphaq} where we show that   \PERFECTC{} is polynomial-time reducible to \QIND{}. To achieve this, to any synchronous game $G$  we associate  an undirected graph $X(G)$  (see  \Def{gamegraph}) and show that $G$  admits  a perfect entangled  strategy if and only  if  the $(X(G),|Q|)$-independent set game has  a perfect strategy (where   $Q$ denotes the question set of $G$).

We note that following the completion of this work it was communicated to us by Ji that building on his recent results in  \cite{Ji13}  he has independently  obtained Result 1. The proof of this fact  has not been published 
but can be derived by appropriately combining the results in \cite{Ji13} together with two additional reductions (that are not stated in \cite{Ji13}). 
Furthermore, in contrast to \cite{Ji13} our approach is constructive and  the final instance of \QIND{} is given explicitly in terms of the instance of \PERFECT.

As was  already mentioned it is  currently not known whether the  entangled value of a nonlocal game (with finite question and answer sets) is always attained by some entangled strategy. In fact, there is evidence that this might not be true. The first  example of a nonlocal game with answer sets of  {\em infinite} cardinality  for which the entangled value is only attained in the limit was identified recently in  \cite{Mancinska14}. 

In this work we  consider the attainability question  restricted to  perfect strategies for symmetric synchronous nonlocal games. A  synchronous game  is called {\em symmetric} if   interchanging the roles of the players does not affect  the value of the Boolean predicate (\cf \Def{symmetric}). We note that all  games of relevance to this work (e.g. homomorphism) are in fact symmetric.
In   \Thm{Attain}  we show that independent set games again  capture the hardness of the attainability question for symmetric games.

\medskip
\noindent\textbf{Result 2:} \textit{Suppose that every  independent set game $G$ satisfying $\omega^*(G)=1$ admits  a perfect entangled strategy. Then the same holds for all symmetric synchronous nonlocal  games. }
\medskip 

To obtain  our second result  we   show that vanishing-error strategies for a symmetric synchronous game $G$ give rise to vanishing-error strategies for an appropriate independent set game defined in terms of the game graph of~$G$.
Notice that since independent set games are synchronous, our second result can be understood as identifying  a class  of synchronous games which captures the hardness of the attainability question for perfect strategies for the entire class of symmetric synchronous nonlocal games. Nevertheless, we note that presently we do not know   whether independent set games satisfy the assumption of Result 2.

A number of interesting results have been derived recently concerning the interplay between the theory of graphs and nonlocal games, e.g.  \cite{Cabello14,Roberson14,Chailloux14, Arkhipov12}.
In   \Sect{GameGraph}  we take a similar approach and  associate  an undirected  graph, called its {\em game graph},  to an arbitrary synchronous nonlocal game. As already mentioned  this is an essential ingredient in showing that \PERFECTC{} is polynomial-time reducible to \QIND{}. In  \Thm{lowerbounds} 
we identify   combinatorial lower bounds on  the classical and entangled value of synchronous nonlocal games in terms of their corresponding game graphs.

\medskip
\noindent\textbf{Result 3:} \textit{ Let $G$ be a synchronous game  with question set $Q$ and uniform distribution of questions. If $X=X(G)$ is  the game  graph of $G$  then  $\omega(G)\ge \big(\al (X)/ |Q|\big)^2,  \text{ and }\ 
  \omega^*(G)\ge \big( \al_p(X)/|Q|\big)^2,$ where $\al(X)$ denotes the independence number of $X$ and $\al_p(X)$ the  projective packing number of $X$ (cf. \Def{projpacking}).}

\section{Preliminaries}
We denote    the set of $d\times d$ Hermitian operators by $\mcs^d$. Throughout this work we equip $\mcs^d$ with the Hilbert-Schmidt inner product  $\la X,Y\ra=\tr(XY^*)$.   An operator  $X\in \mcs^d$ is called {\em positive}, denoted by $X\succeq 0$,  if $\psi^*X\psi\ge 0$ for all $\psi\in \C^d$.  The set of $d\times d$ positive operators is denoted by $\mcs^d_+$. We use the notation  $X\succeq Y$ to indicate that  $X-Y\succeq 0$.  An operator $X$ is called an (orthogonal) {\em projector} if it satisfies $X=X^*=X^2$. 
The {\em support} of an operator $X$, denoted $\supp(X)$, is defined as the projector on the range of $X$. The canonical orthonormal basis of $\C^d$ is denoted by $\set{e_i: i\in [d]}$, where  $[d]:=\set{1,\dotsc,d}$.

\paragraph{Classical strategies and value.} A deterministic strategy for a nonlocal game $G(\pi,V)$ consists  of a pair of functions, $f_A: Q\rightarrow A$ and $f_B: R\rightarrow B$, which the players use in order to determine their answers. The {\em classical value} of the game $G$, denoted by $\omega(G)$, is equal to the maximum expected probability with which the players can win the game using deterministic strategies. Specifically, 
\be
  \omega(G) :=\max \sum_{q\in Q, r\in R} \pi(q,r)V\big(f_A(q),f_B(r)|q,r\big),
\ee 
where the maximization ranges  over all deterministic strategies. 

\paragraph{Quantum strategies and value.}
In this section we   briefly introduce those concepts  from  quantum information theory  that are of relevance  to this work. Readers without the required background  are referred to \cite{NC} for a comprehensive introduction.

To any quantum system \rm{S} we associate a  complex inner product space  $\C^d$,  for some $d\ge 1$. The {\em state space} of the system \rm{S} is defined as the set of  unit vectors in~$\C^d$.  The most basic  way one can extract classical information from a  quantum system ${\rm S}$  is by measuring it.   For the purposes of this paper, the most relevant mathematical formalism of the concept  of a measurement is given by a Positive Operator-Valued Measure (POVM). A POVM is defined in terms of   a family of positive operators  $\mathcal{M}=( M_i\in \mcs^d_+ : i\in [m])$ that sum up to the identity operator, \ie , $\sum_{i\in [m]} M_i=\id_d $. 
According to the axioms of quantum mechanics, if the    measurement $\mathcal{M}$ is performed on a quantum system whose state is given by $\psi\in \C^d$ then the probability that the $i$-th outcome occurs  is given by $\psi^*M_i\psi$. We say that a measurement $\mathcal{M}$  is {\em projective} if all the POVM elements $M_i$ are orthogonal projectors.

Consider two quantum systems ${\rm S_1}$ and ${\rm S_2}$ with corresponding state spaces $\C^{d_1}$ and $ \C^{d_2}$ respectively.   The state space of the joint system $({\rm S_1}, {\rm S_2})$ is given by  $\C^{d_1}\otimes \C^{d_2}$.  Moreover, if ${\rm S_1}$ is in state $\psi_1\in \C^{d_1}$ and ${\rm S_2}$ is in state $\psi_2\in \C^{d_2}$ then the joint system is in state $\psi_1\otimes \psi_2\in \C^{d_1}\otimes \C^{d_2}$. Lastly, if  $(M_i\in \mcs^{d_1}_+ : i\in [m_1])$ and $(N_j\in \mcs^{d_2}_+ : j\in [m_2])$
define  measurements on the individual systems ${\rm S_1}$ and ${\rm S_2}$  then the family of operators $(M_i\otimes N_j\in \mcs^{d_1d_2}_+:i\in [m_1], j\in [m_2])$ defines  a product measurement on the joint system $({\rm S_1}, {\rm S_2})$. 

Given any bipartite quantum state $\psi\in\C^{d}\otimes\C^{d}$, it is possible to choose two orthonormal basis $\set{\alpha_i : i\in[d]}$ and $\set{\beta_i : i\in[d]}$ so that $\psi = \sum_{i=1}^d \lambda_i\, \alpha_i \otimes \beta_i$ and $\lambda_i\ge 0$ for all $i\in[d]$. This is known as the {\em Schmidt decomposition} of $\psi$ and we refer to the $\lambda_i$ as the {\em Schmidt coefficients} of $\psi$. We say that $\psi$ has full Schmidt rank, if all its Schmidt coefficients are positive. We say that $\psi$ is \emph{maximally entangled} if all its Schmidt coefficients are the same. Throughout this paper we use $\phi$ to denote the canonical maximally entangled state  ${1\over \sqrt{d}}\sum_{i=1}^de_i\otimes e_i$ and we make repeated use of the fact that $\phi^*(A\otimes B)\phi={1\over d} \tr(AB\tp)$ for any operators $A,B\in \C^{d\times d}$. 

Consider a nonlocal game $\nong$ with question sets $Q,R$ and answer sets $A,B$ respectively. An   {\em entangled  strategy} for $G$  consists  of a bipartite state    $\psi \in \C^{d_1}\otimes \C^{d_2}$, a POVM $\M_q=(M_{aq}\in \mcs^{d_1}_+: a\in A)$ for each of Alice's questions $q\in Q$  and  a  POVM  $\mathcal{N}_r=(N_{br}\in \mcs^{d_2}_+: b\in B)$ for  each of  Bob's questions $r\in R$. 
Upon receiving questions  $(q,r)\in Q\times R$, Alice performs measurement $\M_q$ on her part of $\psi$ and Bob performs measurement $\mathcal{N}_r$ on his part of $\psi$. The probability that upon receiving questions $(q,r)\in Q\times R$ they answer $(a,b)\in A\times B$ is equal to $\psi^*(M_{aq}\otimes N_{br})\psi$. The {\em entangled value} of  $G$, denoted by $\omega^*(G)$, is the supremum expected  probability with which entangled players can win the game, \ie, 
\be
  \omega^*(G) := \sup \sum_{q\in Q, r\in R}\pi(q,r)
  \sum_{a\in A, b\in B}V(a,b|q,r)
  \psi^*(M_{aq}\otimes N_{br}) \psi,
\ee
where the maximization ranges  over all bipartite quantum states   $\psi \in \C^{d_A}\otimes \C^{d_B}$ and POVMs $(\mathcal{M}_q: q\in Q)$ and $(  \mathcal{N}_r: r\in R)$.  A strategy for $G$ is  called {\em projective} if all the measurements $\mathcal{M}_q$ and $\mathcal{N}_r$ are projective. We say that a nonlocal game $G$ admits  a  perfect quantum strategy if $\omega^*(G)=1$ and  moreover, there exists  a bipartite  state   $\psi \in \C^{d_A}\otimes \C^{d_B}$ and POVMs $(\mathcal{M}_q: q\in Q)$ and $(\mathcal{N}_r: r\in R)$ that achieve this value. 

\paragraph{Graph theory.} 
A graph $X$ is given by an ordered pair of sets $(V(X),E(X))$, where $E(X)$ is a collection of 2-element subsets of $V(X)$. The elements of $V(X)$ are called the {\em vertices} of the graph and the elements of $E(X)$ its {\em edges}. For every edge  $e=\set{u,v}\in E(X)$ we say that $u$ and $v$ are  {\em adjacent} and write $u\sim_X v$ or simply $u\sim v$ if the graph is clear from the context.   
 A set of vertices $S\subseteq V(X)$ is called an  {\em independent set} if no two vertices in $S$ are adjacent. The cardinality of the largest independent set is denoted by $\al(X)$ and is called the {\em independence number} of $X$. 
The {\em complement} of a graph $X$, denoted by $\overline{X}$,  has the same vertex set as $X$, but $u \sim v$ in $\overline{X}$ if and only if $u \ne v$ and $u \not\sim v$ in $X$. A set of vertices $C\subseteq V(X)$ is called a {\em clique} in $X$  if $S$ is an independent set in $\overline{X}$.

\section{Synchronous games}
\label{sec:synchronous}

In  this section we  introduce and study  synchronous  nonlocal games. We first show that     synchronous games can always  be  won with perfect $\pme$ strategies (whenever a perfect strategy exists). Our  main result in this section is \Lem{GtoTG} where we show that any instance of \PERFECT{} is polynomial-time reducible to an instance of \PERFECTC. The main ingredient in this proof is the   notion  of a synchronous extension of a nonlocal game.

\subsection{Definition and basic properties}
\label{sec:CorrDef}

Throughout  this section we focus on games where Alice and Bob share  the same  question and answer sets and furthermore, in order  to win,  they need to give the same answers upon receiving the same questions. 

\begin{definition}\label{def:synchronous} 
A nonlocal  game $G=(V,\pi)$ is called  {\em synchronous} if it satisfies the following properties:
\begin{itemize}
\item[(i)] $A=B$ and $Q=R$;
\item[(ii)] $V(a,b|q,q) = 0, \text{ if } a \ne b$; 
\item[(iii)] for all $q\in Q$, we have $\pi(q,q)>0$.
\end{itemize}
\end{definition}

The notion of synchronous nonlocal games subsumes many classes of nonlocal games that have been recently studied~\cite{Roberson14,Cameron07}. A related concept that has recently been considered is that of synchronous correlations, defined in~\cite{paulsen}. These are correlations (joint conditional probability distributions) such that $\pr(a,a'|q,q) = 0$ whenever $a \ne a'$.

We now study perfect entangled strategies for synchronous games and show that such strategies can, without loss of generality, be assumed to have a certain~form.

\begin{lemma}\label{lem:ProjStrat}
Let $G$ be a synchronous game which admits a perfect entangled  strategy. Then there also exists a perfect {\rm\texttt{PME}} strategy for $G$  where   Bob's projectors are the transpose of  Alice's corresponding  projectors.
\end{lemma}
\begin{proof}
Let  $G$ be a synchronous game   with answer set  $A$ and question set  $Q$. Consider  a  perfect strategy for $G$  given  by a shared state $\psi\in \C^{d_A}\otimes \C^{d_B}$, a  POVM  $\mathcal{M}_q = (M_{aq} : a\in A)$ for each of Alice's questions  and a POVM   $\mathcal{N}_q = (N_{aq} : a\in A)$ for each of Bob's questions. Without loss of generality, we can assume that the shared state is pure and has full Schmidt rank. Let
$  \rho_{aq} := \tr_A\big( (M_{aq} \otimes \id) \si  \big)$ denote  Bob's residual states after Alice has responded $a\in A$ upon receiving question $q\in Q$. 
We first show that 
\be\label{eq:orth}
 \la \rho_{aq}, \rho_{br}\ra=0,   \text{ whenever } V(a,b|q,r) = 0.
\ee
 For this consider a question/answer pair satisfying  $V(a,b|q,r) = 0$ and assume that Bob has received question $r\in Q$. For the players to win, Bob needs to answer $b\in Q$ if he holds the state $\rho_{br}$ since the game is synchronous. On the other hand, he cannot answer $b\in Q$ if he  holds the state $\rho_{aq}$. Since the strategy is perfect, Bob never errs and we can use his answer to perfectly discriminate the states $\rho_{aq}$ and~$\rho_{br}$. Only orthogonal states can be perfectly discriminated and hence we must have that $\la \rho_{aq}, \rho_{br}\ra=0$.

The last step is to use  the support  of Bob's residual states to construct a perfect $\pme$ strategy for~$G$. 
 For all $a\in A$ and $q\in Q$ define $P_{qa} := \supp(\rho_{aq})$.  By definition of $\rho_{aq}$  we have  that $\sum_{a\in A}\rho_{aq}=\tr_A(\psi \psi^*)$ and since $\psi$ has full Schmidt rank it follows that  $\supp(\tr_A(\psi \psi^*))=\supp(\sum_{i=1}^d\lambda_i e_ie_i^*)=I_d$. 
 On the other hand, since  $G$ is a synchronous game, it follows from    \EqRef{orth} that $\la \rho_{aq}, \rho_{a'q}\ra=0$ for $a \neq a'$ and thus  $\supp(\sum_{a\in A}\rho_{aq})=\sum_{a\in A}\supp(\rho_{aq})=\sum_{a\in A} P_{aq}$ for every $q\in Q$. Summarizing we have that $\sum_{a\in A} P_{aq}=I_d$ for all $q\in Q$ and thus   we can define projective measurements $\mathcal{P}_q:=(P_{aq} : a\in A)$ for Alice and $\mathcal{R}_q:=(P_{aq}\tp : a\in A)$ for Bob.
 
 Consider the strategy where the players share the  state   $\phi={1\over \sqrt{d}}\sum_{i=1}^de_i\otimes e_i \in\C^d \otimes \C^d$, Alice uses the projective measurement $\mathcal{P}_q $ upon receiving question $q\in Q$ and Bob uses the projective measurement $\mathcal{R}_q$ upon receiving $q\in Q$. 
To see that this strategy never errs, note that the probability to answer $(a,b)\in A\times A$ upon receiving question pair $(q,r)\in Q\times Q$ is 
  $\pr(a,b|q,r) = 
  \phi^* (P_{aq} \otimes P_{br}\tp) \phi = 
  \frac{1}{d}\tr(P_{aq} P_{br}).$
 Since the supports of orthogonal states are orthogonal it follows from \EqRef{orth} that  $\pr(a,b|q,r)=0$ whenever $V(a,b|q,r)=0$.
\end{proof}

This result was known  for  graph  coloring \cite{Cameron07} and graph homomorphism games \cite{Roberson14}. Since both of these game classes  are synchronous nonlocal games, \Lem{ProjStrat} subsumes  both of these results.

\begin{rem}\label{rem:remark}
Notice that the perfect strategy guaranteed by \Lem{ProjStrat} has the property that $\pr(a,b|q,r)=\pr(b,a|r,q)$ for all $a,b\in A$ and $q,r\in Q$. This observation is  used in  \Lem{alphaq}.
\end{rem}

\subsection{Synchronous extension}
In this section we introduce the notion of the synchronous extension of a nonlocal game (\cf \Def{corrextension}). We also establish that \PERFECT{} is polynomial-time reducible to \PERFECTC{}.

In order to  reduce instances of \PERFECT{} to those of \PERFECTC{}, to any game $G$ we associate  a synchronous   game $\tg$ where  we can also ask Alice any of Bob's questions and vice versa. The winning condition in $\tilde{G}$ is the same as in $G$ if both  players are asked their original questions or the other player's questions. When both players are given the same question, we require that their answers coincide,  therefore ensuring that $\tg$ is synchronous. For simplicity
we assume that the question sets and also the answer sets of the original game $G$ are disjoint. Note however that this is not truly a restriction since any game can be converted into an equivalent game with disjoint question sets and disjoint answer sets, for instance by letting $Q' = \{(q,0) : q \in Q\}$ and $R' = \{(r,1) : r \in R\}$, and similarly for $A$ and $B$.
\begin{definition}\label{def:corrextension}
Let $G$ be a nonlocal game with disjoint question sets $Q,R$  and disjoint answer sets  $A, B$. The {\em synchronous extension} of $G$,  denoted by $\tg$,  is a new  synchronous game  with question and answer sets
$$\tilde{Q} = Q \cup R \quad \& \quad \tilde{A} = A \cup B.$$
The  probability distribution $\tilde{\pi}$ on the question set  $\tilde{Q}\times \tilde{Q}$  is any distribution  of full support\footnote{We could also allow zero probabilities for questions that correspond to zero probability questions    in the original game $G$.}. Lastly  the verification predicate  $\tilde{V}$ is given by:
\be\label{eq:VOrig}
\tilde{V}(a,b | q,r)=\tilde{V}(b,a | r,q)=V(a,b|q,r), \text{ for all } a \in A, \ b \in B, \ q \in Q, \ r\in R,
\ee
\be\label{eq:corre}
\tilde{V}(a,a'| q,q)= \delta_{aa'} \text{ and }  \tilde{V}(b,b' | r,r) = \delta_{bb'}, \text{ for all } q \in Q, \ r \in R, \ a,a' \in A, \ b,b' \in B,
\ee
\be\label{eq:sdvdevger}
\tilde{V}(y,y' | x, x') = 0 \text{ if either } (x,y) \text{ or } (x',y') \text{ is an element of } (R \times A) \cup (Q \times B),
\ee
and it evaluates to one  in all remaining cases.  Notice that  condition \EqRef{VOrig} ensures that players give correct answers upon receiving their original questions or when their  roles are reversed. Furthermore, condition \EqRef{corre} ensures the game is synchronous and \EqRef{sdvdevger} ensures that only Alice's answers are accepted for Alice's questions and only Bob's answers are accepted for Bob's questions. 
 \end{definition}

Generally the synchronous extension might be harder to win than the original game. However, as we will see in the next section,  any perfect $\pme$ strategy for the game $G$, can be also be used to win $\tg$ perfectly.

\subsection{Reducing \PERFECT{} to \PERFECTC}
Using the notion of the synchronous extension we are now ready  to prove the main result in this section. 
\begin{lemma} 
A  nonlocal game $G$ has a perfect {\rm \texttt{PME} }strategy  if and only if its synchronous extension  $\tilde{G}$ has a perfect entangled  strategy. In particular,  {\rm \PERFECT} is polynomial-time reducible to {\rm \PERFECTC}.
\label{lem:GtoTG}
\end{lemma}

\begin{proof}
First, assume that  $G$ has a perfect $\pme$ strategy using a maximally entangled state $\phi\in\C^d\otimes\C^d$ and  projective measurements $\mathcal{P}_q = (P_{aq}: a\in A)$ and $\mathcal{R}_r = (R_{br} : b\in B)$ for Alice and  Bob respectively. Also for all $q\in Q$ and $r\in R$  let $\mathcal{P}_q\tp$ and $\mathcal{R}\tp_{r}$ denote  the projective measurements obtained by taking the transpose of all the projectors within the projective measurements $\mathcal{P}_q$ and $\mathcal{R}_r$ respectively. 
To play the game $\tg$  the players use the following strategy: Alice measures her part of $\phi$ using  $\mathcal{P}_q$ upon receiving question $q \in Q$ and with $\mathcal{R}\tp_r$ upon receiving question $r \in R$. In the former case she responds with some $a \in A$, while in the latter she responds with some  $b \in B$, where $a$ and $b$ are the respective measurement outcomes. Bob acts similarly, except that he uses his original measurements $\mathcal{R}_r$ for a  question $r \in R$ and  $\mathcal{P}\tp_q$ for a question $q \in Q$.

It remains to verify that this defines  a perfect strategy for $\tg$. To do so we show that the players never return   answers for which $\tilde{V}$ evaluates to zero.  First, note that by construction both players only respond with Alice's answers when asked Alice's questions and similarly for Bob's questions and answers. Therefore they never lose due to condition~(\ref{eq:sdvdevger}). Next we will show that condition~(\ref{eq:VOrig}) never causes the players to lose $\tg$. If both players are given questions from their original question sets in $G$, then their strategies are exactly as they were in $G$, and since their strategy for $G$ was perfect they will win in this case. If Alice is given $r \in R$ and Bob is given $q \in Q$, then they will respond with some $b \in B$ and $a \in A$ with probability equal to
\[\phi^* \left(R\tp_{br} \otimes P\tp_{aq}\right)\phi = \frac{1}{d} \tr\left( R\tp_{br}  P_{aq} \right) = \frac{1}{d} \tr\left( P_{aq} R\tp_{br} \right) = \phi^* \left( P_{aq} \otimes R_{br} 
  \right)\phi.\]
This is the probability of Alice and Bob outputting $a$ and $b$ respectively when receiving $q$ and $r$ in the original game $G$. If this probability is greater than 0, then $\tilde{V}(b,a|r,q) = V(a,b|q,r) = 1$ since they win $G$ perfectly. Therefore condition~(\ref{eq:VOrig}) never causes Alice and Bob to lose $\tg$.

Lastly, for all  $q\in Q$ and $a\ne a'\in A$ we have that
\[\pr(a,a'|q,q) = \phi^* \left(P_{aq} \otimes P\tp_{a'q}\right)\phi = {1\over d} \tr(P_{aq}P_{a'q})=0,\]
and similarly for $b \ne b' \in B$ and $r \in R$. Therefore the players always give the same answer when asked the same question and thus they never lose $\tg$ due to condition~(\ref{eq:corre}). Since there are no other ways for the players to lose $\tg$, we have shown that they win this game perfectly.

To show the other direction let us assume $\tg$ has a perfect strategy. By construction $\tg$ is synchronous, hence \Lem{ProjStrat} allows us to conclude that there exists a perfect \pme{} strategy for $\tg$. Since $\tg$ contains the original game $G$, any perfect strategy for $\tg$ can also be used to win $G$ perfectly.
\end{proof}

In fact, the proof of \Lem{GtoTG} shows that any (not necessarily perfect) $\pme$  strategy for $G=(V,\pi)$  can be used to win $\tg=(\tilde{V},\tilde{\pi})$ with at least as high probability  of success if $\tilde{\pi}|_G = \pi$. Here, we have used $\tilde{\pi}|_G$ to refer to the distribution obtained from $\tilde{\pi}$ by restricting to questions in $G$ and re-normalizing.

\section{Game graphs}
\label{sec:GameGraph}
In this section we introduce the notion of the game graph of a synchronous game (\cf \Def{gamegraph}).  Our main result in this section  is \Thm{Packing} where  we  relate the existence of perfect entangled  strategies for a synchronous game to  the projective packing number of its  game graph. 
This is  used in \Sect{pertoqind} to reduce  \PERFECTC{}  to \QIND{}. Lastly, in  \Thm{lowerbounds} we  identify  combinatorial lower bounds on  the classical and entangled values of synchronous games in terms of their game graphs.

\subsection{Definition and some properties}
A  nonlocal game $G=(V,\pi)$ admits a perfect entangled strategy if there exist a quantum state 
$\psi \in \C^{d_A}\otimes \C^{d_B}$ and POVM measurements   $(\mathcal{M}_{qa}: a\in A)\subseteq  \mcs_+^{d_A}$ and $(\mathcal{N}_{rb}: b\in B)\subseteq  \mcs_+^{d_B}$ such that 
\be\label{eq:perfect}
  \psi^*(M_{qa}\otimes N_{rb})\psi=0, \text{ when } V(a,b|q,r)=0 
  \text{ and } \pi(q,r)>0.
\ee

We have already seen  in  \Lem{ProjStrat} that a  synchronous game has  a perfect entangled strategy  if and only if it has a perfect $\pme$ strategy.  This implies that for synchronous games Condition~\EqRef{perfect}  reduces   to a set of orthogonality relations  between the measurement operators. Next, for every synchronous nonlocal game we associate an undirected graph which encodes these required orthogonalities as adjacencies.

\begin{definition}\label{def:gamegraph}
Let $G$  be a synchronous game with question set $Q$ and answer set $A$.  The {\em game graph} of $G$, denoted $X(G)$, is the undirected graph with vertex set $A \times Q$ where  $(a,q)$ is adjacent to $(a',q')$ if $V(a,a'|q,q')=0$ or $V(a',a|q',q)=0$.
\end{definition}

An important feature of game graphs is that their vertex set admits a natural partition into cliques. Specifically, for a given question $q\in Q$ of a synchronous game $G$, the vertices of $V_q := \set{(a,q) : a \in A}$ are pairwise adjacent in $X(G)$. This observation will be important for   the proofs in this section.

\subsection{Synchronous games and the projective packing number}
\label{sec:Packing}
In this section we  show that a synchronous game admits a perfect entangled strategy if and only if its game graph has a projective packing of  value  $|Q|$ (\cf  \Thm{Packing}).

We first recall  the definition of the projective packing number of  a graph \cite{Roberson14,robersonthesis}. 
\begin{definition}\label{def:projpacking}
A {\em $d$-dimensional projective packing} of a graph $X=(V,E)$ consists of  an assignment of projectors $P_u \in \mcs^d_+$ to every vertex  $u\in V$ such that
\be
  \tr(P_u P_v)=0,  \text{ whenever } u\sim_X v.
\ee
The {\em value}  of a projective packing using  projectors  $P_u \in \mcs^d_+$ is defined as 
\be
  \frac{1}{d} \sum_{u\in V} \tr(P_u).
\ee
The {\em projective packing numbe}r of a graph $X$, denoted $\alpha_p(X)$, is defined as  the supremum of the values over  all projective packings of the graph  $X$.
\end{definition}
Notice  that the supremum in the definition of projective packing number is necessary because it is not clear that  $\alpha_p(X)$ is always attained by some projective packing of the graph  $X$.  We now give an upper bound on the projective packing number of a  game graph.

\begin{lemma}\label{lem:bound}
For  any synchronous game $G$ with question set $Q$ we have that  $\alpha_p\big(X(G)\big) \le |Q|.$
\end{lemma}

\begin{proof}
Let $(P_{aq}: a \in A, q \in Q)$ be a $d$-dimensional projective packing of $X(G)$. The vertices in $V_q=\set{(a,q) : a \in A}$ are pairwise adjacent and thus the projectors $P_{aq}$  are pairwise orthogonal for every $q\in Q$. Therefore,
\[\sum_{a\in A} \tr(P_{aq}) = \sum_{a\in A} \rank(P_{aq}) \le d,\]
where $\rank(M)$ is the rank of matrix $M$. From the above inequality we further obtain that
\[\frac{1}{d} \sum_{(a,q)\in A\times Q} \tr(P_{aq}) = \frac{1}{d}\sum_{q\in Q} \sum_{a\in A} \tr(P_{aq}) \le \frac{1}{d}|Q|\cdot d = |Q|,\]
and thus   $\alpha_p\big(X(G)\big) \le |Q|$.
\end{proof}

In view of \Lem{bound}  it is natural to ask when it is  the case that  $\alpha_p\big(X(G)\big) = |Q|$. As it turns out  this happens exactly when there exists  a perfect entangled  strategy for $G$.

\begin{theorem}
Let $G$ be a synchronous game with question set $Q$. Then $G$ has a perfect entangled  strategy if and only if its game graph  has a projective packing of value~$\abs{Q}$.
\label{thm:Packing}
\end{theorem}

\begin{proof}
Let  $G$ be  a synchronous game with a perfect entangled  strategy. By  \Lem{ProjStrat}, there exists a perfect projective strategy for $G$ that uses maximally entangled state $\phi\in\C^d\otimes\C^d$, where Alice's and Bob's projectors are transpose to each other. Let $P_{aq}\in \mcs^d_+$ be Alice's projector associated with question $q\in Q$ and answer $a\in A$. Since  this  strategy is perfect   we have that
\be\label{eq:whathever}
  0 = \phi^*(P_{aq} \otimes P_{a'q'}\tp)\phi=
    \frac{1}{d} \tr(P_{aq} P_{a'q'}),
\ee
whenever $V(a,a'|q,q')=0$ or $V(a',a|q',q)=0$. It follows immediately from \Eq{whathever} that   the projectors $P_{aq}$ form   a $d$-dimensional projective packing of $X(G)$. Since $ \sum_{a\in A} P_{aq} = \id_d$ it follows that  
\be
  \frac{1}{d} \sum_{(a,q)\in A \times Q} \tr(P_{aq}) = 
  \frac{1}{d} \sum_{q\in Q} \tr(\id_d) = \abs{Q},
\ee
and thus  value of this packing is $\abs{Q}$. Lastly, by \Lem{bound}  we get that $\alpha_p\big(X(G)\big)= |Q|$.

For the other direction, assume that $X(G)$ has a $d$-dimensional projective packing $(P_{aq}: a\in A, q\in Q)$ of value $\abs{Q}$. Since $G$ is a synchronous game we have that $(q,a) \sim (q,a')$ for $a\neq a'\in A$ and $q\in Q$. This implies  that $\sum_{a\in A} P_{aq} \preceq \id_d$, as the added projectors are mutually orthogonal. Furthermore, since  the value of the  projective packing is $\abs{Q}$, we obtain
\be
  \abs{Q} = \frac{1}{d} \sum_{(a,q)\in A \times Q} \tr(P_{aq}) = 
  \sum_{q\in Q} \Big(\frac{1}{d} 
  \tr\big( \sum_{a\in A} P_{aq} \big) 
  \Big) \leq 
  \sum_{q\in Q} \frac{1}{d}  \tr(I_d) \leq \abs{Q},
\label{eq:Packing}
\ee
and thus  \Eq{Packing} holds throughout with  equality.
In particular, 
 $\tr\big( \sum_{a\in A} P_{aq} \big)=\tr(I_d) $, and since  $\sum_{a\in A} P_{aq} \preceq \id_d$ we conclude that $\sum_{a\in A} P_{aq} = \id_d$ and thus $\mathcal{P}_q = (P_{aq}: a\in A)$ forms a  valid projective measurement. By the definition of the edge set  of  $X(G)$, we see that Alice and Bob can win with probability one, if they measure a maximally entangled state using projective measurements $\mathcal{P}_q^{\phantom{T}}$ and $\mathcal{P}_{q}\tp$ respectively.
\end{proof}

\subsection{Lower bounding $\omega(G)$  and $\omega^*(G)$  for  synchronous games}

In this section   we derive combinatorial lower bounds on  the classical and entangled values of synchronous nonlocal  games in terms of the independence number and the projective packing number of their game graphs respectively (\cf \Thm{lowerbounds}).

Our first result gives     a necessary and sufficient condition for  
the existence of a perfect classical strategy. 

\begin{lemma}\label{lem:gergergre}
Let $G$ be a synchronous game with question set $Q$  and let $X:=X(G)$ be the its  game graph. Then,  $G$ has a perfect classical strategy if and only if  $\alpha(X)=|Q|.$
\end{lemma}
\begin{proof}Let $f_A,f_B: Q\rightarrow  A$  be a  perfect deterministic strategy for the game $G$. Since $G$ is synchronous  we have that $f_A=f_B=:f$. Set  $V_q=\set{ (a,q): a\in A}$ and notice that $\set{ V_q: q\in Q}$ forms a clique cover of $X$ of cardinality $|Q|$.  This shows that $\al(X)\le |Q|$. Lastly, we show that   $S=\set{ (q,f(q)): q\in Q}$ is an independent set in $X$. Indeed,  since $f$ is a perfect strategy,  for any   $(q,f(q)), (r,f(r))\in S$ we have  that $V(f(q),f(r)|q,r)=V(f(r),f(q)|r,q)=1$.  This implies that $(q,f(q))\not \sim  (r,f(r))$. 

Conversely, let $S$ be an independent set  in   $X$ of cardinality $|Q|$. Since $\set{ V_q: q\in Q}$ is a clique cover of cardinality $|Q|$, for every $q\in Q$, the intersection $S\cap V_q$ contains exactly one vertex of $X$ which we denote by   $(q,a_q)$.  Define   $f: Q\rightarrow A$  where $f(q)=a_q$ for every $q\in Q$ and consider the deterministic strategy for  $G$ where both players determine  their answers  using   $f$. It remains to show that this is a perfect classical  strategy. Assume for contradiction that  there exist $q,r\in Q$ such that $V(f(q),f(r)|q,r)=0$. By definition of $X$ this implies that $(f(q),q)\sim (f(r), r)$,  contradicting the fact that $S$ is an independent  set in $X$. 
\end{proof}

As an immediate consequence of \Lem{gergergre} we  recover the well-known fact that there exist a graph homomorphism from a graph $X$ to a graph $Y$ if and only if $\alpha(X\ltimes Y)=|V(X)|$. Here $ X\ltimes Y$ denotes the {\em homomorphic product} of $X$ and $Y$ whose vertex set is given by $V(X)\times V(Y)$ and $(x,y)\sim (x',y')$ if and only if $[(x=x')$ and $y\ne y']$ or $[x\sim x'$ and $y\not\sim y']$.  To recover this result  from \Lem{gergergre} notice  that the game graph for  the $(X,Y)$-homomorphism game is given precisely  by   $X\ltimes Y$ (see also \cite{Roberson14}). 

We now  proceed to lower bound the classical and entangled values of  synchronous games. 

\begin{theorem}\label{thm:lowerbounds}
 Consider a synchronous game  $G$ with question set $Q$ and uniform distribution of questions. If $X=X(G)$ is  the game  graph of $G$  then, 
\be\label{thm:qlb}
\omega(G)\ge \big(\al (X)/ |Q|\big)^2 \text{ and }\ 
  \omega^*(G)\ge \big( \al_p(X)/|Q|\big)^2. 
\ee
\end{theorem}
\begin{proof} First, we consider the classical case. 
Our goal is to exhibit  a deterministic  strategy that wins on at least   $\al(X)^2$ out of the $|Q|^2$ pairs of possible  questions.  Let $S$ be  an independent  set in $X$ of cardinality $\al (X)$.  By definition of the edge set of $X$,   for any pair $(a, q), (b, r)\in S$  we have that
\begin{equation}\label{eq:sdbeddss}
  V(a,b|q,r)=1 \text{ and }  V(b,a|r,q)=1.
\end{equation}
Set $Q'=\set{q\in Q : \exists a\in A \text{ such that } (a,q)\in S}$.  Since $G$ is synchronous and  $S$ is an independent  set,  for  every $q\in Q'$ there exists a unique $a\in A$ such that $(a,q)\in S$, which we denote by $f(q)$.  Furthermore, notice    that $|Q'|=\al(X)$.   
Consider  the following   deterministic strategy: If a player receives as question an element  $q\in Q'$ he responds with $f(q)$ and if $q\not \in Q'$ his answer is arbitrary.  
It follows from   \EqRef{sdbeddss} that  for $q,r\in Q'$,   the players win when asked $(q,r)$ and $(r,q)$. Since  $|Q'|=\al(X)$ this strategy is correct on at least $\al(X)^2$ of the $|Q|^2$ possible questions. 

Next we consider the entangled case. 
Let $(P_{aq}: a\in A, q\in Q)$ be a $d$-dimensional projective packing for~$X$ of value $\gamma$, \ie,  
$  
\gamma={ 1\over d}\sum_{ a\in A, q\in Q} \tr(P_{aq}). 
  $ We construct an entangled  strategy whose  value is at least  ${ \gamma^2/|Q|^2}.$  Recall  that  for all  $q\in Q$ the set $V_q=\set{ (a,q) : a \in A} $ forms  a clique in $X$.  This implies that for fixed $q\in Q$ and  $a\ne a'\in A$  the projectors $P_{aq}$ and $ P_{a'q} $ are pairwise orthogonal  and thus $\sum_{a\in A} P_{aq}\preceq  \id_d$.
Consider the following entangled strategy for $G$: The players share the maximally entangled state $\phi \in \mathbb{C}^d\otimes \mathbb{C}^d  $ and  for every $q\in Q$  Alice uses the projective measurement $(P_{aq}: a\in A) \cup( I-\sum_{a\in A}  P_{aq})$ and Bob uses the measurement 
  $(P_{aq}\tp: a\in A) \cup( I-\sum_{a\in A}  P_{aq}\tp )$.
Using this strategy the players win with probability at least 
\begin{equation}\label{eq:sdbee}
  { 1\over |Q|^2} \sum_{a,b, q,r}  
  \phi^*(P_{aq} \otimes  P_{br}\tp) \phi \; V(a,b|q,r)=
  { 1\over d|Q|^2} \sum_{a,b,q,r} \tr(P_{aq} P_{br})V(a,b|q,r).
\end{equation}
  If $V(a,b|q,r)=0$  then $(a,q)\sim (b,r)$ in the game graph and by definition of the projective packing we have that $\tr(P_{aq}P_{br})=0$. Thus  \EqRef{sdbee} gives that 
\begin{equation}\label{eq:1}
  \omega^*(G)\ge { 1\over d|Q|^2} \sum_{a,b,q,r} \tr(P_{aq} P_{br})= { 1\over d|Q|^2}  \tr(P^2),
  \end{equation}
  where  $P=\sum_{a\in A,q\in Q} P_{aq}$.
  By  the Cauchy-Schwartz inequality we get   that  $\tr(P^2)\ge  \tr(P)^2/ d$.
Finally, since $\gamma=\tr(P)/ d,$ it follows from \EqRef{1} that  $\omega^*(G)\ge \gamma^2/ |Q|^2$
and the proof is completed. 
\end{proof}

If a synchronous game $G$ satisfies $\al_p(X(G))=|Q|$, it follows   from   \Thm{lowerbounds} that  $\omega^*(G)=1$. On the other hand  we have seen in  \Thm{Packing}  that if there exists a projective packing for the game graph  with value equal to $|Q|$ then  $G$  has a perfect quantum strategy. Notice that  these two conditions are not equivalent    since we could have $\al_p(X(G))=|Q|$ without this value being attained.

\section{Independent set games}
\label{sec:Alphaq}
In this section we show  that \PERFECTC{}  is polynomial-time reducible to \QIND{} (\cf \Lem{alphaq}). This fact combined with the reduction of  \PERFECT{} to \PERFECTC{} derived in \Lem{GtoTG}   implies that \PERFECT{} is polynomial-time reducible to \QIND, which is the  main result in this paper. Additionally  we  consider   the attainability problem for perfect strategies and synchronous games. 
In    \Thm{Attain}  we show that   if any independent set game whose entangled value is one also admits a perfect strategy then the same is true for all symmetric  synchronous games.

\subsection{Reducing \PERFECT{} to \QIND}
\label{sec:pertoqind}

Recall that in the  $(X,t)$-independent set game  the players try  to convince a verifier that the graph $X$ contains  an independent set of size $t$. The verifier selects uniformly at random  $(i,j)\in [t]\times [t]$ and sends $i$ to Alice and $j$ to Bob. The players respond with vertices $u,v\in V(X)$ respectively.  The verification predicate evaluates to zero in the following three cases: $[i=j \text{ and } u\ne v]$ or $[i\ne j \text{ and } u=v]$ or $[i\ne j \text{ and } u\sim_X v]$.

The independence number of a graph $X$ can equivalently be  defined  as the largest integer $t\ge 1$ for which the $(X,t)$-independent set game admits a perfect classical strategy. Similarly,  the {\em quantum independence number} of a graph $X$, denoted by  $\alpha_q(X)$ is defined as the largest integer $t\ge1$ for which the $(X,t)$-independent set  game  admits a perfect entangled strategy~\cite{Roberson14}.

It is known that  the projective packing number is an upper bound to the
quantum independence number.  
\begin{lemma}\cite[6.11.1]{robersonthesis}
\label{lem:AlphaqToPacking}
Let $X$ be a graph and  $k\in\N$.
If $\alpha_q(X) \ge k$ then there   exists a projective packing of $X$ with value $k$. In particular,  $\alpha_q(X) \le \alpha_p(X).$
\end{lemma}

We are now ready to prove the main result in this section.

\begin{lemma}\label{lem:alphaq}
Let $G$ be a synchronous game with question set $Q$.  Then 
$G$ has a perfect entangled  strategy if and only if $\alpha_q(X(G)) = |Q|.$ In particular, 
{\rm \PERFECTC} is polynomial-time reducible to {\rm \QIND}.
\end{lemma}
\begin{proof}
Assume first   there exists a perfect entangled strategy for $G$. By \Lem{ProjStrat} there also exists a perfect \pme{} strategy $S$ for $G$ where Bob's projectors are the transpose of Alice's corresponding projectors. For all $a,a'\in A$ and $q,q'\in Q$ let  $\pr(a,a'|q,q')$ be the probability that Alice and Bob answer $a$ and $b$ respectively upon receiving questions $q$ and $q'$ when employing strategy $S$. We now  construct a perfect strategy for  the $(X(G),|Q|)$-independent set game. Without loss of generality, we may assume that  $Q$ itself is used as the question  set. Consider the following strategy: upon receiving $q$ and $q'$ respectively, Alice and Bob use strategy $S$ to obtain answers $a$ and $a'$. They then output vertices $(a,q)$ and $(a',q')$ respectively. 
If $\pr(a,a'|q,q') > 0$ then by Remark~\ref{rem:remark} we have that also  $\pr(a',a|q',q) > 0$. From this we see that both $V(a,a'|q,q') = 1$ and $V(a',a|q',q) = 1$, since $S$ was perfect. This implies that $(a,q)$ and $(a',q')$ are (possibly equal) nonadjacent vertices in $X(G)$. If $q \ne q'$, then these two vertices are not equal and are therefore distinct nonadjacent vertices of $X(G)$, as required by the independent set game. If $q = q'$, then since $G$ is a synchronous game and $S$ is perfect, we have that $a = a'$ and therefore the two outputted vertices are equal as required. This shows that using this strategy allows Alice and Bob to win the independent set game perfectly and thus $\alpha_q(X(G)) \ge |Q|$. On the other hand  from \Lem{AlphaqToPacking} together with \Lem{bound} it follows that $\alpha_q(X(G)) \le |Q|$ and thus $\alpha_q(X(G)) = |Q|$.

Conversely,  assume   that $\alpha_q(X(G)) = |Q|$. By \Lem{AlphaqToPacking}  and \Lem{bound}  there exists a projective packing of $X(G)$ of value $|Q|$, and therefore by \Thm{Packing}  there exists a perfect entangled strategy for $G$.
\end{proof}

Lastly, combining \Lem{GtoTG} and \Lem{alphaq} directly yields the main result of this paper.
\begin{theorem}\label{thm:main2}
A nonlocal game $G$ with question sets $Q$ and $R$ admits a perfect \pme{} strategy if and only if 
\be
  \alpha_q\big(X(\tilde{G})\big) = \abs{Q} + \abs{R}.
\ee
In particular, {\rm \PERFECT} is polynomial-time reducible to {\rm \QIND}. 
\end{theorem}

\subsection{Attainability problem for perfect strategies}
 In this section we focus on the  the attainability problem for perfect strategies and  show that     the attainability question  for  symmetric synchronous games reduces to the attainability  question for independent set games. 

\begin{definition}\label{def:symmetric}
A synchronous game $G=(V,\pi)$ is called {\em symmetric} if 
$$V(a,a'|q,q')=V(a',a|q',q), \ \text{ for all } a,a'\in A, q,q'\in Q.$$
\end{definition}
Notice that all  synchronous games we consider in this work  (\eg, homomorphism and coloring games) are~symmetric. 
\begin{theorem}
Suppose that any independent set game $G$ satisfying $\omega^*(G)=1$ admits  a perfect entangled strategy. Then the same holds for all symmetric synchronous nonlocal  games.
\label{thm:Attain}
\end{theorem}
\begin{proof}
Let $G=(V,\pi)$ be any symmetric synchronous game with question set $Q$ and answer set $A$. Assume that  $\omega^*(G) = 1$ and let $X := X(G)$ be its game graph. Define  $G'=(V',\pi)$ to be the  $(X, \abs{Q})$-independent set game with $\pi$ as the distribution of questions.   The crux  of the proof is  that  from any strategy $S$  that succeeds in $G$  with probability at least $1-\varepsilon$, we can construct a  strategy $S'$  that wins  $G'$ with   probability at least $1-\varepsilon$.  
Similarly to \Lem{alphaq}, using the  strategy $S$ for $G$ we define the following strategy $S'$ for $G'$:  Upon receiving $q\in Q$, Alice uses strategy $S$ for $G$   and obtains an answer $a\in A$. She  then replies  with vertex $(q,a)$ of $X$. Similarly,  Bob, upon receiving $q'\in Q$ he uses strategy $S$ for $G$ to obtain answer $a'\in A$. He then replies with vertex $(q',a')$ of $X$. Let $\pr_S(a,a'|q,q')$ denote  the probability that using strategy $S$ the players respond with $(a,a')\in A\times A$ upon receiving questions $q,q'\in Q$ respectively. 
By assumption  we have that 
\be\label{eq:winninprobabilitys}
\begin{aligned}
\omega^*(G,S):=\sum_{q\in Q, a\in A}\pi(q,q)\pr_S(a,a|q,q)+ & \sum_{q\ne q'\in Q}\pi(q,q')\sum_{a,a'\in A: V(a,a'|q,q')=1}\pr_S(a,a'|q,q')\ge 1-\varepsilon.
\end{aligned}
\ee
Furthermore, by definition of  the strategy $S'$  we have that 
\be\label{eq:probabilitites}
\pr_S(a,a'|q,q')=\pr_{S'}\big((q,a),(q',a')|q,q'\big), \text{ for all } q,q'\in Q.
\ee
 Since $G'$ is an independent set game  we have  $V'\big((q,a),(q',a')|q,q'
\big)=1$ if and only if  
\be\label{eq:predicatev}
[q=q' \text{ and } a=a']  \text{ or }  [q\ne q' \text{ and } (a,q)\not\sim_X(a',q')].\ee
Since the game $G$ is symmetric, Condition~\EqRef{predicatev} is equivalent to 
\be\label{eq:winningcond}
[q=q' \text{ and } a=a']  \text{ or }  [q\ne q' \text{ and } V(a,a'|q,q')=1].
\ee
Combining \EqRef{probabilitites} with  \EqRef{winningcond}, and the fact that $V(a,a'|q,q) = 1 \Rightarrow a = a'$, it follows that the probability of winning the game $G'$ using  strategy $S'$ is at least 
$\omega^*(G,S) \ge 1- \varepsilon$. 
Since $\omega^*(G) = 1$ this argument can be repeated for any  $\varepsilon$  arbitrarily close to 0 which implies that the  entangled  value of $G'$ is equal to one. By the assumption of the theorem, this implies that there is a perfect quantum strategy for $G'$ and thus by \Lem{alphaq} there exists a perfect quantum strategy for $G$. 
\end{proof}

\paragraph{Acknowledgements.} D.~E.~Roberson and A.~Varvitsiotis are supported in part by the Singapore National Research Foundation under NRF RF Award No. NRF-NRFF2013-13. L.~Man\v{c}inska is supported by the Singapore Ministry of Education under the Tier 3 grant MOE2012-T3-1-009.

\bibliographystyle{alphaurl}
\newcommand{\etalchar}[1]{$^{#1}$}

\end{document}